\documentclass[10pt]{article}
\usepackage{graphicx, verbatim, url, amssymb, amsmath, amsfonts, amsthm, latexsym, eurosym, epigraph}

\newtheorem{example}{Example}

\newtheorem{theorem}{Theorem}

\begin{document}

\title{A theory on power in networks}

\author{Enrico Bozzo \\
Department of Physics, Mathematics, and Computer Science \\ 
University of Udine \\
\url{enrico.bozzo@uniud.it} \and 
Massimo Franceschet\\
Department of Physics, Mathematics, and Computer Science \\ 
University of Udine \\
\url{massimo.franceschet@uniud.it}}

\maketitle

\epigraph{It was an exceedingly long lay that, indeed; and though from the magnitude of the figure it might at first deceive a landsman, yet the slightest consideration will show that though seven hundred and seventy-seven is a pretty large number, yet, when you come to make a teenth of it, you will then see, I say, that the seven hundred and seventy-seventh part of a farthing is a good deal less than seven hundred and seventy-seven gold doubloons; and so I thought at the time.}{\textit{Moby-Dick, or The Whale  \\ Herman Melville}}

\begin{abstract}
The eigenvector centrality equation $\lambda x = A \, x$ is a successful compromise between simplicity and expressivity. It claims that central actors are those connected with central others. For at least 70 years, this equation has been explored in disparate contexts, including econometrics, sociometry, bibliometrics, Web information retrieval, and network science. We propose an equally elegant counterpart: the \textit{power equation} $x = A x^{\div}$, where $x^{\div}$ is the vector whose entries are the reciprocal of those of $x$. It asserts that power is in the hands of those connected with powerless others. It is meaningful, for instance, in bargaining situations, where it is advantageous to be connected to those who have few options. We tell the parallel, mostly unexplored story of this intriguing equation.
\end{abstract}

\section{A portrait of power} \label{portrait}
A \textit{network} consists of a crowd of actors and a set of binary relations that tie pairs of actors.
Networks are pervasive in the real world. Nature, society, information, and technology are supported by ostensibly different networks that in fact share an amazing number of interesting structural properties.

Networks are modelled in mathematics as \textit{graphs}: actors are represented as points (also called nodes or vertices) and relations are depicted as lines (also called edges or arcs) connecting pairs of points.
In this work we focus on undirected graphs, where the edges do not have a particular orientation.
A meaningful question on networks is the following: which are the most central (important) nodes?
Many  measures have been proposed to address this question. Among them, \textit{eigenvector centrality}\footnote{In this paper we will refer to eigenvector centrality simply as centrality.} states that
an actor is central if it is connected with central actors.
This circular definition is captured by an elegant recursive equation:
\begin{equation}
\lambda x = A x
\label{eq:centrality}
\end{equation}
\noindent
where $x$ is a vector containing the sought centralities, $A$ is a matrix encoding the network, and $\lambda$ is a positive constant.  Two actors in a network that are tied by an edge are said to be neighbors. Equation \ref{eq:centrality} claims two important properties of centrality:
(i) the centrality of an actor is directly correlated with the number of its neighbors; and (ii) the centrality of an actor is directly correlated with the centrality of its neighbors. Central actors are those with many ties or, for equal number of ties, central actors are those who are connected with central others. This intriguing definition has been discovered and rediscovered many times in different contexts. It has been investigated, in chronological order, in econometrics, sociometry, bibliometrics, Web information retrieval, and network science (see \cite{F11-CACM} for an historical overview).

In some circumstances, however, centrality -- the quality of being connected to central ones -- has limited utility in predicting the locus of \textit{power} in networks \cite{E62,CY92,B87}.
Consider exchange networks, where the relationship in the network involves the transfer of valued items (i.e., information, time, money, energy). A set of exchange relations is positive if exchange in one relation \textit{promotes} exchange in others and negative if exchange in one relation \textit{inhibits} exchange in others \cite{CEGY83}.
In \textit{negative exchange networks}, power comes from being connected to those who have few options. Being connected to those who have many possibilities reduces one's power.
Think, for instance, to a social network in which time is the exchanged value.
Imagine that every actor has a limited time to listen to others and that each actor divides its time between its neighbors.
Clearly, exchange of time in one relation precludes the exchange of the same time in other relations. What are the actors that receive most attention?
These are the nodes that are connected to many neighbors with few options, since they receive almost full attention from all their neighbors.
On the other hand, actors connected to few neighbors with a lot of possibilities receive little consideration, since their neighbors are mostly busy with others. 

In this paper we propose a theory on power in the context of networks. We start by the following thesis:
An actor is powerful if it is connected with powerless actors. We implement this circular thesis with the following equation:
\begin{equation}
x = A \, x^{\div}
\label{eq:power}
\end{equation}
\noindent where $x$ is the sought power vector, $A$ is a matrix encoding the network and $x^{\div}$ is the vector whose entries are the reciprocal of those of $x$. Equation \ref{eq:power} states two important properties of power: (i) the power of an actor is directly correlated with the number of its neighbors; and (ii) the power of an actor is inversely correlated with the power of its neighbors. The first property seems reasonable: the more ties an actor has, the more powerful the actor is. The second property characterizes power: for equal number of ties, actors that are linked to powerless others are powerful; on the other hand, actors that are tied to powerful others are powerless.

We investigate the existence and uniqueness of a solution for Equation \ref{eq:power} exploiting well-known results in combinatorial matrix theory. We study how to regain the solution when it does not exist, by perturbing the matrix representing the network. We formally relate the introduced notion of power with alternative ones and empirically compare them on the European natural gas pipeline network.

\section{Motivating example} \label{motivation}

In his seminal work on power-dependance relations, dated 1962, Richard Emerson claims that power is a property of the social relation, not an attribute of the person: ``X has power'' is vacant, unless we specify ``over whom''. Power resides implicitly in other's dependance, and dependance of an actor A upon actor B is (i) directly proportional to A's motivational investment in goals mediated by B, and (ii) inversely proportional to the availability of those goals to A outside the A--B relation. The availability fo such goals outside of the relation refers to alternatives avenues of goal-achievement, most notably other social relations \cite{E62}. This type of relational power is endogenous with respect to the network structures, meaning that it is a function of the position of the node in the network. Exogenous factors, such as allure or charisma, which are external to the network structure,
might be added to endogenous power to complete the picture.

We begin with some small, archetypal examples typically used in exchange network theory to informally illustrate the notion of power and sometimes to distinguish it from the intersecting concept of centrality \cite{EK10}. Consider a 2-node path:
 
$$A - B$$

The situation is perfectly symmetric and  a reasonable prediction is that both actors have the same power. In a 3-node path
 
$$A - B - C$$

\noindent much is changed. Intuitively, B is powerful, and A and C are not. Indeed, both A and C have no alternative venues besides B (both depend on B), while B can exclude one on them by choosing the other.\footnote{We assume here the so called 1-exchange rule, meaning that each node may exchange with at most 1 neighbor. Equivalently, we consider negative exchange network, in which the exchange in one relation inhibits exchange in others.} A different basis for B's power might be implicit in the economic principle of satiation (see Chapter 12 in \cite{EK10}).
In a 4-node path
 
$$A - B - C - D$$

\noindent actors B and C hold power, while A and D are dependent on either B or C. Nevertheless, the power of B is less here than in the 3-node path: in both cases, A depends on B, but in the 3-node path, also C depends on B, while in the 4-node path C has an alternative (node D). Hence, node B is less powerful in the 4-node path with respect to the 3-node path since its neighbors are more powerful.   
Finally, the 5-node path
 
$$A - B - C - D - E$$

\noindent is interesting since it discriminates power from centrality. All traditional central measures (eigenvector, closeness, betweenness) claim that C is the central one. Nevertheless, B and D are reasonably the powerful ones. Again, this because they negotiate with weak partners (A and C or E and C), while C bargains with strong parties (B and D). This example is useful to illustrate an additional subtle aspect of power. Notice that in both the 5-node path and the 4-node path B is surrounded by nodes (A and C) that are locally similar (for instance, they have the same degree in both paths). However, the power of C is reasonably less in the 5-node path than in the 4-node path, hence we might expect that the power of B is higher in the 5-node path with respect to the 4-node path. This separation is only possible if the notion of power spans beyond the local neighborhood of a node (for example, if power is recursively defined).

\begin{figure}[t]
\begin{center}
\includegraphics[scale=0.60]{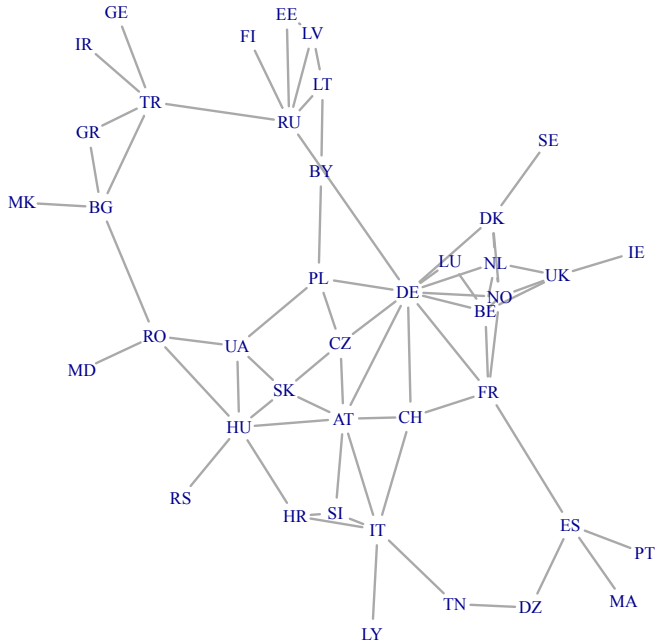}
\end{center}
\vspace{-2cm}
\caption{The European natural gas pipeline network.}
\label{fig:gas}
\end{figure}

As a larger and more realistic example, consider Figure \ref{fig:gas}, which depicts the European natural gas pipeline network. Nodes are European countries (country codes according to ISO 3166-1) and there is an undirected edge between two nations if there exists a natural gas pipeline that crosses the borders of the two countries. Data has been downloaded from the website of the International Energy Agency (\url{www.iea.org}). The original data corresponds to a directed, weighted multigraph, with edge weights corresponding to the maximum flow of the pipeline. We simplified and symmetrized the network, mapping the edge weights in a consistent way. 

This is a negative exchange network, since the exchange of gas with a country precludes the exchange of the very same gas with others. Intuitively, powerful countries are those connected with states that have few possibilities to exchange the gas. Suppose that country B is connected to countries A and C, and B is the only connection for them: $A - B - C$. Countries A and C can sell or buy gas only from B, while country B can chose between A and C. Reasonably, the bargaining power of B is higher, and this traduces in higher revenues or less expenses for B in the gas negotiation.

\section{A theory on power} \label{theory}

Let $G$ be an undirected, weighted graph. The graph $G$ may contain \textit{loops}, that are edges from a node to itself. The edges of $G$ are labelled with positive weights. Let $A$  be the adjacency matrix of $G$, that is $A_{i,j}$ is the weight of edge $(i,j)$ if such edge exists, and $A_{i,j} = 0$ otherwise. Hence $A$ is a square, symmetric, non-negative matrix. Loops in $G$ correspond to elements in the main diagonal of $A$.

The \textit{centrality problem} is as follows: find a vector $x$ with positive entries such that
\begin{equation}\label{centrality}
\lambda x = A \, x.
\end{equation}
where $\lambda > 0$ is a constant. This means that $\lambda x_i = \sum_j A_{i,j} \, x_j$, that is, the centrality of a node is proportional to the weighted sum of centralities of its neighbors.
Notably, this  is the main idea behind PageRank, Google's original Web page ranking algorithm. PageRank 
determines the importance of a Web page in terms of the importance assigned to the pages hyperlinking to it. Besides Web information retrieval, this thesis has been successfully exploited
in disparate contexts, including bibliometrics, sociometry, and econometrics \cite{F11-CACM}.

We define the \textit{power problem} as follows: find a vector $x$ with positive entries such that
\begin{equation}\label{power}
x = A \, x^{\div}.
\end{equation}
where we denote with $x^{\div}$ the vector whose entries are the reciprocals of those of $x$. 
This means that $x_i = \sum_j A_{i,j} / x_j$, that is, the power of a node is equal to the weighted sum of reciprocals of power of its neighbors. Notice that if $\lambda x = A \, x^{\div}$, then, setting $y = \sqrt{\lambda} \, x$, we have that $y = A \, y^{\div}$, hence the proportionality constant $\lambda$ is not necessary in the power equation.
This notion of power is relevant on negative exchange networks \cite{CY92,B87}. On this networks, when a value is exchanged between actors along a relation, it is consumed and cannot be exchanged along another relation. Hence, important actors are those in contact with many actors having few exchanging possibilities. 

Finally, the \textit{balancing problem} is the following: find a diagonal matrix $D$ with positive main diagonal such that
$$S = DAD$$
is doubly stochastic, that is, all rows and columns of $S$ sum to $1$. The balancing problem is a fundamental question that is claimed to have first been used in the 1930's for calculating traffic flow \cite{BCP93} and since then it has been applied in many disparate contexts \cite{KR13}.

It turns out that the power problem is intimately related to the balancing problem.
Given a vector $x$, let $D_x$ be the diagonal matrix whose diagonal entries coincide with those of  $x$. We have the following result.

\begin{theorem} \label{theorem:balancing}
The vector $x$ is a solution of the power problem if and only if the diagonal matrix $D_{x^{\div}}$ is a solution for the balancing problem.
\end{theorem}

\begin{proof}
If $DAD$ is doubly stochastic, then $DAD e =  e$ and $e^T DAD= e^T$, where $e$ is a vector of all $1$s. Actually, since $A$ and $D$ are symmetric, it holds  that $DAD e =  e \Leftrightarrow e^T DAD= e^T$.
If the vector $x$ does not have zero entries, then $D_x$ is invertible and $D^{-1}_{x} = D_{x^{\div}}$. We have that:
$x = A \, x^{\div} \Leftrightarrow   D_x e = A D_{x^{\div}} e \Leftrightarrow e = D^{-1}_x A D_{x^{\div}} e   \Leftrightarrow e = D_{x^{\div}} A D_{x^{\div}} e.$
\end{proof}

\subsection{Existence and unicity of a solution}

The link between the balancing and the power problem that we established in Theorem \ref{theorem:balancing} allows us to investigate a solution of the power problem (Equation \ref{power}) using the well-sedimented theory of matrix balancing. 

We recall that \textit{diagonal} of a square $n \times n$ matrix is a sequence of $n$ elements that lie on different rows and columns of the matrix. A permutation matrix is a square $n \times n$ matrix that has exactly one entry $1$ in each row and each column and $0$s elsewhere. Clearly, each diagonal corresponds to a permutation matrix where the positions of the diagonal elements correspond to those of the unity entries of the permutation matrix.
In particular, the identity matrix $I$ is a permutation matrix and the diagonal of $A$ associated with  $I$ is called the main diagonal of $A$. A diagonal is positive if all its elements are greater than $0$. A matrix $A$ is said to have \textit{support} if it contains a positive diagonal, and it is said to have \textit{total support} if $A \neq 0$ and every positive element of $A$ lies on a positive diagonal. Clearly, total support implies support. 

A matrix is \textit{indecomposable} (\textit{irreducible}) if it is not possible to find a permutation matrix $P$ such that $$P^TAP = \begin{pmatrix} X  &  Y \cr 0 & Z  \end{pmatrix}$$
where $X$ and $Z$ are both square matrices and $0$ is a matrix of $0$s, otherwise $A$ is \textit{decomposable} (\textit{reducible}).
A matrix is \textit{fully indecomposable} if it is not possible to find permutation matrices $P$ and $Q$ such that
$$PAQ = \begin{pmatrix} X  &  Y \cr 0 & Z  \end{pmatrix}$$
where $X$ and $Z$ are both square matrices, otherwise $A$ is \textit{partly decomposable}. Clearly, a matrix that is fully indecomposable is also irreducible. It also holds that full indecomposability implies total support \cite{Br80}. Moreover, the adjacency matrix of a bipartite graph is never fully indecomposable, while the adjacency matrix of a non-bipartite graph is fully indecomposable if and only if it has total support and is irreducible \cite{CD72}. We say that a graph has support, total support, is irreducible, and is fully indecomposable if the corresponding adjacency matrix has these properties.

\begin{figure}[t]
\begin{center}
\includegraphics[scale=0.40, angle=270]{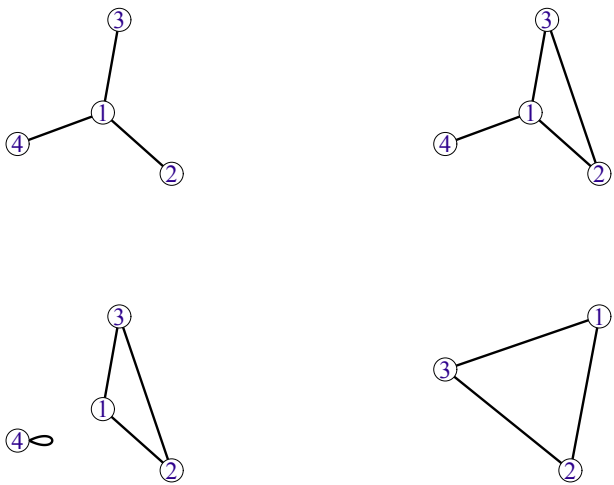}
\end{center}
\caption{Top-left: the graph has no support, since a spanning cycle forest is missing. Top-right: the graph has support formed by edges $(1,4)$ and $(2,3)$ but the support is not total (edges $(1,3)$ and $(1,2)$ are not part of any spanning cycle graph). Bottom-left: the graph has total support but it is not irreducible, hence it is not fully indecomposable. Bottom-right: the graph has total support, is irreducible and not-bipartite, hence it is fully indecomposable.}
\label{4graphs}
\end{figure}

The combinatorial notions presented above are rather terse. Fortunately, most of them have a simple interpretation in graph theory. It is known that irreducibility of the adjacency matrix corresponds to connectedness of the graph. Moreover, given an undirected graph $G$, let us define a \textit{spanning cycle forest} of $G$ a spanning subgraph of $G$ whose connected components are single edges or cycles (including loops, that are cycles of length 1). It is easy to realize that there exists a correspondance between diagonals in the adjacency matrix and spanning cycle forests in the graph. 
Hence a graph has support if and only if it contains a spanning cycle forest and it has total support if and only if each edge is included in a spanning cycle forest. Four examples are depicted in Figure \ref{4graphs}.

The following is a well-known necessary and sufficient condition for the solution of the balancing problem \cite{SK67,CD72}.

\begin{theorem} \label{theorem:total-support}
Let $A$ be a symmetric non-negative square matrix. A necessary and sufficient condition for the existence of a doubly stochastic matrix $S$ of the form $DAD$, where $D$ is a diagonal matrix with positive main diagonal, is that $A$ has total support. If $S$ exists, then it is unique. If $A$ is fully indecomposable, then matrix $D$ is unique.
\end{theorem}

It follows that the power problem $x = A \, x^{\div}$ has a solution exactly on the class of graphs that have total support. Moreover, if the graph is fully indecomposable, then the solution is also unique.

\subsection{Perturbation: regaining the solution}

What about the power problem on graphs whose adjacency matrix has not total support? For such graphs, the power problem has no solution. Nevertheless, a solution can be regained by perturbing the adjacency matrix of the graph in a suitable way. We investigate two perturbations on the adjacency matrix $A$:

\begin{enumerate}
\item \textit{Diagonal perturbation}: $A^{D}_{\alpha} = A + \alpha I$, where $\alpha > 0$ is a damping parameter and $I$ is the identity matrix;
\item \textit{Full perturbation}: $A^{F}_{\alpha} = A + \alpha E$, where $\alpha > 0$ is a damping parameter and $E$ is a full matrix of all $1$s.
\end{enumerate}

Matrix $A^{F}_{\alpha}$ is clearly fully indecomposable, it has total support and is irreducible. Hence, the power problem (as well as the centrality one) on a fully perturbed matrix has a unique solution. On the other hand, matrix $A^{D}_{\alpha}$ has total support. Indeed, if $A_{i,j} > 0$ and $i=j$, then the main diagonal $A_{k,k}$, for $1 \leq k \leq n$ is positive and contains $A_{i,j}$. If $i \neq j$, then the diagonal $A_{i,j}, A_{j,i}, A_{k,k}$, for $1 \leq k \leq n$ and $k \neq i,j$ is positive and contains $A_{i,j}$. 
Thus, the power problem on a diagonally perturbed matrix has a solution. Moreover, the solution is unique if $A$ is irreducible, since it is known that for a symmetric matrix $A$ it holds that $A$ is irreducible if and only if $A+I$ is fully indecomposable \cite{BR91}.
Interestingly enough, the diagonal perturbation, besides providing convergence of the method, is useful to incorporate exogenous power in the model. By setting a positive value in the $i$-th position of the diagonal, we are saying that node $i$ has a minimal amount of power that is not a function of the position of the node in the network. Hence, we can play with the diagonal of the adjacency matrix to assign nodes with potentially different entry levels of exogenous power.

Intuitively, the diagonal perturbation is less invasive than its full counterpart: the former modifies the diagonal elements only, the latter touches all matrix elements. However, how much invasive is the perturbation with respect to the resulting power? To investigate this issue, we computed the  correlation between original and perturbed power solutions. A simple and intuitive measure of the correlation between two rankings of size $n$ is Kendall rank correlation coefficient $k$, which is the difference between the fraction of concordant pairs $c$ (the number of concordant pairs divided by $n (n-1)/2$) and that of discordant pairs $d$ in the two rankings: $k = c - d$. Since $c + d = 1$, we have that the probability that two random pairs in the two rankings are concordant is $c = (k + 1) / 2$ and that for discordant pairs is $d = (-k + 1) / 2$. The coefficient runs from -1 to 1, with negative values indicating negative correlation, positive values indicating positive correlation, and values close to 0 indicating independence. We used the following network data sets: a social network among dolphins  \cite{LN04}, the Madrid train bombing terrorist network \cite{H06}, a social network of Jazz musicians \cite{GD03}, a network of friendships between members of a karate club \cite{Z77}, a collaboration network of scholars in the field of network science \cite{NG04}, and a co-appearance network of characters in the novel Anna Karenina by Lev Tolstoj \cite{K93}.

The main outcomes of the current experiment are as follows (see Figure \ref{exp2-compare}): (i) as soon as the damping parameter is small, both diagonal and full perturbations do not significantly change the original power; (ii) power with diagonal perturbation is closer to original power than power with full perturbation; (iii) the larger the damping parameter, the lower the adherence of perturbed solutions to the original one. 

\begin{figure*}[t]
\begin{center}
\includegraphics[scale=0.30, angle=270]{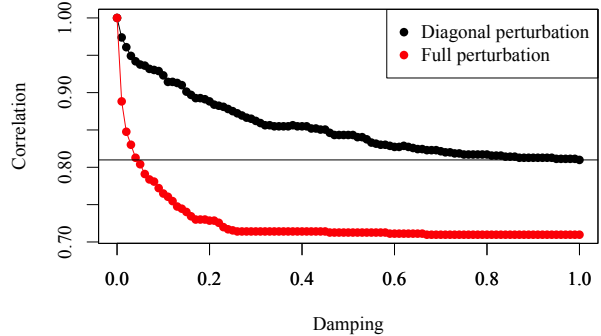}
\end{center}
\caption{Correlation between original and perturbed powers varying the damping parameter from 0 to 1 on the largest biconnected component of the social network among dolphins (which has total support). The horizontal line corresponds to the correlation with diagonal perturbation and maximum damping. The correlation on the other networks is similar.}
\label{exp2-compare}
\end{figure*}

\subsection{Computing power}

Due  to the established relationship between the balancing problem and the power problem, we can use known methods for the former in order to solve the latter. The simplest approach to solve the Equation \ref{power} is to set up the iterative method
\begin{equation}\label{SKiteration}
x_{k+1}=A\, x_k^{\div},
\end{equation}
known as Sinkhorn-Knopp method \cite{SK67}.
If we set $x_0 = e$, the vector of all $1$s, then the first iteration $x_1 = A e$,
that is $x_1(i) = \sum_j A_{i,j}$ is the degree $d_i$ of $i$. The second iteration $x_2 = A \,  (A e)^{\div}$, that is $x_2(i) = \sum_j A_{i,j} / d_i$ is the sum of reciprocals of the degrees of the neighbors of $i$.

If $A$ has total support then the Sinkhorn-Knopp method converges (more precisely, the even and odd iterates of the method converge to power vectors that differ by a multiplicative constant). The convergence is linear with rate of convergence that depends of the subdominant eigenvalue of the the balanced matrix $S = DAD$ (see Theorem \ref{theorem:total-support}) \cite{SK67}. 
In some cases, however, the convergence can be very slow. In \cite{KR13} the authors propose a faster algorithm based on Newton's method that we now describe according to our setting and notations. In order to solve Equation \ref{power}, we apply Newton's method for finding the zeros of the function $f:\mathbb{R}^n\rightarrow \mathbb{R}^n$ defined by $f(x)=x-Ax^{\div}$. It is not difficult to check that
$$\frac{\partial f_i}{\partial x_j}(x)=\delta_{i,j}+\frac{A_{ij}}{x_j^2}, \qquad i,j=1,\ldots,n,$$
where $\delta_{i,j}=1$ if $i=j$ and $\delta_{i,j}=0$ otherwise.
We collect these partial derivatives in the Jacobian matrix of $f$ that turns to be
$$J_f(x)=I+AD_{(x^2)^{\div}},$$
where the squaring of $x$ is to be intended entrywise. Formally, the Newton method applied to the equation $f(x)=0$ becomes
\begin{eqnarray*}
x_{k+1}&=&x_k- J_f^{-1}(x_k)f(x_k)\\
       &=&J_f^{-1}(x_k)\bigl(J_f(x_k)x_k-f(x_k)\bigr)\\
       &=&J_f^{-1}(x_k)\bigl(x_k+AD_{(x_k^2)^{\div}}x_k-x_k+Ax_k^{\div}\bigr)\\
       &=&2J_f^{-1}(x_k)Ax_k^{\div}.
\end{eqnarray*}

To apply Newton's method exactly it is necessary to solve a linear system at each step and this would be too much expensive. Nevertheless an approximate solution of the system obtained by means of an iterative method is sufficient, giving rise to an inner-outer iteration. 
This approach becomes appealing when the matrix that has to be balanced is symmetric and sparse, which is the case for the power problem on real networks \cite{KR13}.

We experimentally assessed the complexity of computation of power on the above mentioned real social networks (in fact, we used the largest biconnected component for the first three networks in order to work also with totally supported graphs). We use both Sinkhorn-Knopp method (SK for short) and Newton method (NM for short). We consider the computation on the original matrix as well as on the perturbed ones. We use as a benchmark the complexity of the computation of centrality using the power method (PM for short). The complexity is expressed as the overall number of matrix-vector product operations. If a matrix is sparse (this is the case for all tested networks), such operation has linear complexity in the number of nodes of the graph. The main empirical findings are summarized as follows (see Table \ref{statistics}): (i) SK on the original matrix is significantly slower than PM and diagonal perturbation does not significantly change its speed; (ii)
full perturbation significantly increases the speed of SK so that the complexity of SK with full perturbation and that of PM are comparable; moreover, the larger the damping parameter, the faster the method;
(iii) NM on the original matrix is much faster than SK: its complexity is comparable to that of fully perturbed SK and PM, and (iv) NM with diagonal perturbation is even faster than NM and the larger the damping parameter, the faster the method.

 \begin{table}[t]
\begin{center}
\begin{tabular*}{1\textwidth}{@{\extracolsep{\fill}}lllllll}
\textbf{Network} & \textbf{PM} & \textbf{SK} & \textbf{SK-D} & \textbf{SK-F} & \textbf{NM} & \textbf{NM-D} \\ \hline
Dolphin       & 73 & 294  & 300 & 72 &  47 & 30 \\ \hline
Madrid        & 28 & 416  & 320 & 78 &  46 & 27 \\ \hline
Jazz          & 42 & 300  & 288 & 78 &  37 & 27 \\ \hline
Karate        & 42 & --   & 494 & 52 &  -- & 31 \\ \hline
Collaboration & 65 & --   & 9740 & 30 & -- & 33  \\ \hline
Karenina      & 24 & --   & 1006 & 32 & -- & 32 \\ \hline
\end{tabular*}
\end{center}
\caption{Complexity of computation of power with different methods: PM (benchmark), SK (SK without perturbations for totally supported networks), SK-D (SK with diagonal perturbation and damping 0.15), SK-F (SK with full perturbation and damping 0.01), NM (NM without perturbations for totally supported networks), NM-D (NM with diagonal perturbation and damping 0.15). \label{statistics}}
\end{table}

\subsection{Relationship with alternative power measures} \label{alternatives}

Bonacich has proposed a family of parametric measures depending on two parameters $\alpha$ and $\beta$ \cite{B87}. If $A$ is the adjacency matrix of the graph, the Bonacich index $x$ is defined as:
\begin{equation}\label{bona}
x = \alpha Ae +\beta Ax.
\end{equation}

The index for a node is the sum of two components: a first one (weighted by the parameter $\alpha$) depends on the node's degree, a second one (weighted by the parameter $\beta$) depends on the index on the node's neighbors. From Equation \ref{bona}, under the condition that $I-\beta A$ is not singular,  it is possible to obtain the following explicit representation of the proposed measure:
\begin{equation}\label{bonae}
x = \alpha(I-\beta A)^{-1}Ae = \alpha (\sum_{k=0}^{\infty} \beta^k A^{k+1})  e.
\end{equation}

The equivalence with the infinite sum holds when $|\beta| < 1/r$, where $r = \max_i |\lambda_i|$, with $\lambda_i$ the eigenvalues of $A$ (that is, $r$ is the  spectral radius of $A$).  When the parameter $\beta$ is positive, the index is a centrality measure. In particular, the measure approaches eigenvector centrality as a limit as $\beta$ approaches $1/r$. On the other hand, when $\beta$ is negative, the index is a power measure: it corresponds to a weighted sum of odd-length paths (with positive sign) and even-length paths (with negative sign) \cite{B87}. Hence, powerful nodes correspond to nodes with many powerless neighbors.  Finally, when $\beta = 0$, the measure boils down to degree centrality.

The difficulty with this measure is that it is parametric, that is, it depends on parameters $\alpha$ and $\beta$. While it is simple to set the parameter $\alpha$, and in particular it can be used to assign exogenous power to nodes, the choice for the parameter $\beta$ in more delicate. In particular, the index makes sense when the parameter $|\beta| < 1 / r$, hence the spectral radius $r$ must be computed or at least approximated.

The precise relationship between Bonacich power (Bonacich index with negative $\beta$) and power defined in Equation \ref{power} is explained as follows. If we set $x_0 = (1/\gamma) e$ in Newton's iteration  for the computation of power described above we obtain
$$x_1=2\gamma (I+\gamma^2 A)^{-1}Ae.$$
But this first approximation is a member of the family of Bonacich's measures with $\alpha=2\gamma$ and $\beta=-\gamma^2$. Since $\beta$ is negative, we are indeed facing  a measure of power. Hence, Bonacich power can be considered as a first-order approximation of power using Newton method.

Bozzo et al.\ have investigated power measures on \textit{sets} of nodes \cite{BFR15}. Given a node set $T$ let $B(T)$ be the set of nodes whose neighbors all belong to $T$. Notice that nodes in $B(T)$ do not have connections outside $T$, hence are potentially at the mercy of nodes in $T$.
The authors define a power function $p$ such that $p(T) =  |B(T)| - |T|$. Hence, a set $T$ is powerful if it has potential control over a much larger set of neighbors $B(T)$. The power measure is interpreted as the characteristic function of a coalition game played on the graph and the \textit{Shapley value} of the game, that is, the average marginal contribution to power carried by a node when it is added to any node set, is proposed as a measure of power for single nodes. Interestingly enough, the discovered game-theoretic power measure corresponds to the second iteration of Sinkhorn-Knopp method for the computation of power as defined by Equation \ref{power}, that is, to the sum of reciprocals of neighbors' degrees.

The study of power has a long history in economics (in its acceptation of bargaining power) and sociology (in its interpretation of social power). Consider the most basic case where just two actors $A$ and $B$ are involved in a negotiation over how to divide one unit of money. Each party has an alternate option -- a backup amount that it can collect in case negotiations fails, say $\alpha$ for $A$ and $\beta$ for $B$. A natural prediction, known as Nash's bargaining solution \cite{N50}, is that the two actors will split the surplus $s = 1 - \alpha - \beta$, if any, equally between them: if $s < 0$ no agreement among $A$ and $B$ is possible, since any division is worse than the backup option for at least one of the party. On the other hand, if $s >= 0$, then $A$ and $B$ will agree on $\alpha + s/2$ for $A$ and $\beta + s/2$ for $B$. 

A natural extension of the Nash bargaining solution from pairs of actors to \textit{networks} of actors has been proposed in \cite{R84,CY92} and further investigated in particular in \cite{KT08,BBCKA15}. In the following, we describe the dynamics that captures such an extension. Let $A$ be the adjacency matrix of an undirected, unweighted graph $G$. Hence $A_{i,j} = 1$ if there is an edge $(i,j)$ in $G$ and $A_{i,j} = 0$ otherwise. Negotiation among actors is only possible along edges, each pair of actors on an edge negotiate for a fixed amount of 1\euro{}, and each actor may conclude a negotiation with at most 1 neighbor (1-exchange rule). For every edge $(i,j)$, define:

\begin{itemize}
\item  $R_{i,j}$ as the amount of \textit{revenue} actor $i$ receives in a negotiation with $j$;  
\item $L_{i,j}$ as the amount of revenue actor $i$ receives in the \textit{best alternative} negotiation excluding the one with $j$.
\end{itemize}

Notice that matrices $R$ and $L$ have the same zero-nonzero pattern of $A$. More precisely, consider the following iterative process. We start with $R^{(0)}_{i,j} = 1/2$ for all edges $(i,j)$ and $R^{(0)}_{i,j} = 0$ elsewhere. Let $N(i)$ be the set of neighbors of node $i$. For $t > 0$, the best alternative matrix $L^{(t)}$ at time $t$ is:

\begin{equation*}
  L^{(t)}_{i,j} = \left\{
    \begin{array}{ll}
      \max_{k \in N(i) \setminus j} R^{(t-1)}_{i,k} & \text{if } A_{i,j} = 1,\\
      0 & \text{otherwise }
    \end{array} \right.
\end{equation*}

Let the surplus $s^{(t)}_{i,j} =  1 - L^{(t)}_{i,j} - L^{(t)}_{j,i}$ be the amount for which actors $i$ and $j$ will negotiate at time $t$; notice that actor $i$ will never accept an offer from $j$ less than his alternate option $L^{(t)}_{i,j}$ and actor $j$ will never accept an offer from $i$ less than her alternate option $L^{(t)}_{j,i}$. Then the profit matrix $R^{(t)}$ at time $t$ is:

\begin{equation*}
  R^{(t)}_{i,j} = \left\{
    \begin{array}{ll}
      L^{(t)}_{i,j} + s^{(t)}_{i,j}/2 & \text{if } A_{i,j} = 1 \text{ and } s^{(t)}_{i,j} \geq 0,\\
      1 - L^{(t)}_{j,i} & \text{if } A_{i,j} = 1 \text{ and } s^{(t)}_{i,j} < 0,\\
      0 & \text{otherwise }
      
    \end{array} \right.
\end{equation*}

Notice that $R^{(t)}_{i,j} + R^{(t)}_{j,i} = 1$. That is, $R^{(t)}_{i,j}$ and $R^{(t)}_{j,i}$ is the Nash's bargaining solution of a negotiation between actors $i$ and $j$ given their alternate options $L^{(t)}_{i,j}$ and $L^{(t)}_{j,i}$.  
Let $R$ be the fixpoint of the iterative process $R^{(t)}$ for growing time $t$. The \textit{Nash power} $x_i$ of node $i$ is the best revenue of actor $i$ among its neighbors, that is $$x_i = \max_j R_{i,j}$$

Among many other attractive results, \cite{BBCKA15} shows that the dynamics always converge to a fixpoint solution.
Nash power bears some analogy with the one we propose and investigate in the present work, in particular both notions share the same recursive powerful-is-linked-with-powerless philosophy. Nash power for an actor $i$ directly depends on the revenues of $i$ among its neighbors, which directly depend on the alternate options of $i$ among its neighbors, which inversely depends on the revenues of neighbors of $i$, which determine the power of neighbors of $i$. Hence, power of an actor somewhat inversely depends on power of its neighbors. 

Using Kendall correlation, we assessed the overlapping of power, as defined in this paper, with centrality, degree, as well as Bonacich power (that is, Bonacich index with negative parameter $\beta$), Shapley power (that is, the sum of reciprocals of neighbors' degrees), and Nash power on the above mentioned real social networks. 
The main empirical outcomes are summarized in the following (see Table \ref{exp3.correlation}): (i) as expected, both power and centrality are positively correlated with degree, but power is negatively correlated with centrality when the effect of degree is excluded (we used partial correlation); (ii) power is positively correlated with Bonacich power and the association increases as the parameter $\beta$ declines below $0$ down to $-1/r$, with $r$ the spectral radius of the adjacency graph matrix; moreover, the association is higher when the adjacency matrix is perturbed; (iii) power is positively correlated with Shapley power and the association is generally stronger than with Bonacich power;
(iv) power is positively correlated with Nash bargaining network power, but the strength of the correlation is generally weaker than with Shapley and Bonacich power. In particular, we noticed that the Nash-based method maps the power scores of the nodes of the surveyed networks into a small set of values, with very high frequency for values close to 0, 0.5 and 1. Hence, it is difficult to discriminate different gradations of power for nodes. 

\begin{table}[t]
\begin{center}
\begin{tabular}{llllll}
\textbf{Network} & \textbf{D} & \textbf{C} & \textbf{B} & \textbf{S} & \textbf{N} \\ \hline
Dolphin       & 0.81 & 0.35 &  0.89 & 0.91 & 0.72 \\ \hline
Madrid        & 0.62 & 0.33 &  0.69 & 0.68 & 0.48 \\ \hline
Jazz          & 0.85 & 0.62 &  0.91 & 0.85 & 0.17\\ \hline
Karate        & 0.77 & 0.36 &  0.74 & 0.96 & 0.75 \\ \hline
Collaboration & 0.77 & 0.05 &  0.77 & 0.85 & 0.60 \\ \hline
Karenina      & 0.75 & 0.45 &  0.62 & 0.89 & 0.86 \\ \hline
\end{tabular}
\end{center}
\caption{Correlation of power as defined in this paper with degree (D), centrality (C), Bonacich power (B), Shapley power (S) and Nash power (N). For the computation of power, we used diagonal perturbation (damping 0.15). For Bonacich power we used $\alpha = 1$ and $\beta = -0.85/r$, where $r$ is the spectral radius of the graph. \label{exp3.correlation}}
\end{table}

\section{Motivating example reloaded} \label{reload}

In the following we reload examples provided in Section \ref{motivation} and use them as a benchmark to compare the different notions of power described in Section \ref{alternatives}. When the graphs are not totally supported (all the cases but the 2-node path), we used diagonal perturbation with damping 0.15 to obtain a solution. Moreover, we set Bonacich index parameters $\alpha = 1$ and $\beta = -0.85/r$, where $r$ is the spectra radius of the graph.

In the 2-node path all methods agree to give identical power to both nodes. In the 3-node path $A - B - C$ all methods agree that B is the powerful one. Notably, Nash power assigns all power (1) to B and no power (0) to A and C, while the other methods say that A and C hold a small amount of power.
In the 4-node path $A - B - C - D$ all methods claim that B and C are the powerful ones. Moreover, all methods recognize that the power of B in this instance is less than its power in the 3-node path. Finally, in the 5-node path  $A - B - C - D - D$, all methods discriminate B and D as the most powerful nodes, followed by C, and finally A and E, with the only exception of Nash power, which assigns all power (1) to B and D, and null power (0) to all other nodes (hence the central node C has the same power as the peripheral nodes A and E, according to this method). All methods, with the exception of Shapley, notice that the power of B is higher in the 5-node path with respect to the 4-node path. This because Shapley is a local method, while the other ones are global (recursive) methods.

\begin{figure}[t]
\begin{center}
\includegraphics[scale=0.60]{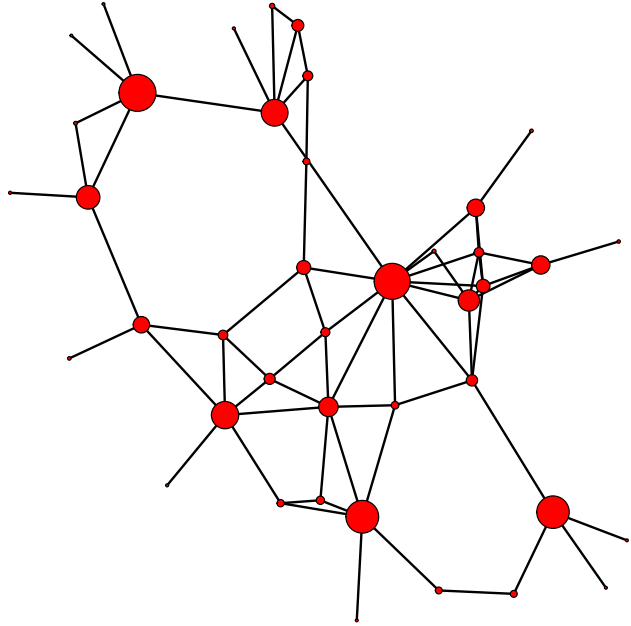}
\end{center}
\vspace{-2cm}
\caption{The European natural gas pipeline network with node size proportional to power.}
\label{fig:gas2}
\end{figure}

Let us now reload the natural gas pipeline example. Figure \ref{fig:gas2} depicts the network with node size proportional to node power. Notice how the power of a node balances the reciprocal of power of its neighbors. We ranked all countries according to the alternative power measures.  Below, we show the Kendall correlation matrix for the following measures: Shapley power (S), Bonacich power (B), power as defined in this paper (P), Nash power (B), and eigenvector centrality (C):

\begin{equation*}
\begin{matrix}
&     \mathbf{S}  &  \mathbf{B}  &  \mathbf{P}  & \mathbf{N} & \mathbf{C} \\
\mathbf{S} &  1.00 &  0.82 &  0.90 &  0.69 &  0.41 \\
\mathbf{B} &  0.82 & 1.00 & 0.84 & 0.61 & 0.46 \\
\mathbf{P}  & 0.90  & 0.84 & 1.00 & 0.72 & 0.47 \\
\mathbf{N} & 0.69  & 0.61 & 0.72 & 1.00 & 0.36 \\
\mathbf{C} & 0.41 & 0.46 & 0.47 & 0.36 & 1.00
\end{matrix} 
\end{equation*}

As expected, P is well correlated with its approximations B and S. Moreover, P is positively correlated with N, but the correlation strength is weeker. Also, the association between P and C is positive but week, and this is mostly explained by the association with degree of both measures. Indeed, if we exclude the effect of degree, this correlation is negative, meaning that the variables are negatively associated when the effect of degree is filtered out. 

These associations are mirrored in the top-10 rankings and ratings listed in Table \ref{tab:rank}. Notice how Germany (DE) is both powerful and central; Italy (IT) is powerful but not central; Norway (NO) is central but not powerful, and there are many countries that are neither powerful nor central outside the rankings. Italy is the archetype of powerful country that is not central: it contracts with nations that are both powerless and peripheral, namely Austria, Switzerland, Croatia, Tunisia, Libya and Slovenia, with only Austria included in the top-10 power list and only Austria and Switzerland included in the top-10 centrality list (and not in the first positions). The ranking according to Nash power is somewhat weird if compared with the other power measures; for instance, Germany has bargaining power 0.5 and is only in 14th postion, tied in the standings with other 12 countries. It is fair noticing that the generalized Nash bargaining solution has been originally proposed in the context of assignments problems (eg., matching of apartments to tenants or students to colleges), and was not suggested as a rating and ranking method for nodes in a network. For instance, in the balanced matching over the gas network Italy preferably negotiates with Libya and Turkey with Georgia. In fact, \cite{CY92} proposes to use the negotiation values obtained by each node in such a solution as a structural power measure (see also Chapter 12 in \cite{EK10} for a similar interpretation). According to the experiments we made for this paper, this interpretation seems to be opinable, but further investigation is necessarily to gain a solid conclusion.

\begin{table}[t]
\begin{center}
\begin{tabular}{lllllllllll}
\textbf{P} & TR & DE & IT & ES & HU & RU & BG & BE & AT & UK \\
& 6.26 & 6.09 & 5.54 & 5.50 & 4.62 & 4.53 & 3.99 & 3.60 & 3.29 & 3.09 \\\\
\textbf{B} &  DE &  IT &  HU &  TR &  AT &  RU &  ES &  BE &  NO &  BG \\
& 7.07 & 4.58 & 4.06 & 3.73 & 3.41 & 3.37 & 3.29 & 3.23 & 2.92 & 2.76 \\\\ 
\textbf{S}  &  TR & ES & IT & DE & RU & HU & BG & RO & UK & AT  \\
& 2.92  &  2.70  &  2.56  &  2.54  &  2.46  &  2.23  &  1.95  &  1.67  &  1.53  &  1.51 \\\\
\textbf{N} &  ES & TR & BG & RU & IT & HU & UK & RO & DK & LV  \\
& 1.00  &  1.00  &  1.00  &  0.87  &  0.83  &  0.83  &  0.75  &  0.75  &  0.75  &  0.75  \\\\
\textbf{C} &  DE &  NO &  BE &  NL &  FR &  AT &  DK &  CH &  CZ &  UK \\
& 1.00 & 0.71 & 0.68 & 0.62 & 0.56 & 0.52 & 0.46 & 0.45 & 0.40 & 0.39  \\
\end{tabular}
\end{center}
\caption{The top-ten powerful and central countries in the European natural gas exchange network.}
\label{tab:rank}
\end{table}

\section{Discussion} \label{discussion}

We have proposed a theory on power in the context of networks. The philosophy underlying our notion of power claims that an actor is powerful if it is connected with many powerless actors. This thesis has its roots and applications mainly in sociology and economics, and this traces an historical parallel with its celebrated linear counterpart, namely eigenvector centrality. The thesis is compiled into an elegant power equation $$x = A \, x^{\div},$$ where $x$ is power and $A$ represents the network. We have investigated the domain of existence and uniqueness of the power equation using results of combinatorial matrix theory. The power equation has a solution exactly on the class of graphs that exhibit total support. Matrix perturbations of the original network can be applied when the network is outside the existence domain. The power equation can be solved using fast iterative methods. Finally, we tested and compared our notion of power on the natural gas pipeline network among European nations. 

The virtues of our definition of power are: (i) it is a simple, elegant, and understandable measure; (ii) it is theoretically well-grounded and directly related to the well-studied balancing problem, making it possible to borrow results and techniques from this setting; (iii) the formulation is not parametric; (iv) it is global (the power of a node depends on the entire network) and can be approximated with a simple local measure - the sum of reciprocals of node degrees - which has a game-theoretic interpretation and can be efficiently computed on all networks. The vices are: (i) an exact solution exists only on the class of totally supported networks;  (ii) it is not immediately normalizable, hence power values for nodes on different networks must be compared with care.


\end{document}